\RequirePackage{fix-cm}
\documentclass[smallextended]{svjour3}       % onecolumn (second format)
\smartqed  % flush right qed marks, e.g. at end of proof
\usepackage{graphicx}
\usepackage{txfonts}
\usepackage{mathdots}
\usepackage{booktabs}

 \journalname{  }

\usepackage{amssymb}
\usepackage{CJK}
\usepackage{mathrsfs}
\usepackage{epsf,epsfig,amsfonts,amsgen,amstext,amsbsy,amsopn,latexsym}

\begin{document}
\title{Construction of EAQECCs with imperfect ebits\\ \thanks{This work is supported by Natural Science Foundation
of China under Grant No.U21A20428.} }


\author{Guanmin Guo  \and Ruihu Li  %etc.
}

\institute{Guanmin Guo$^{*}$ \at \email{gmguo$\_$xjtukgd@yeah.net}   \\
          Ruihu Li \at \email{llzsy110@126.com}  \\
\at Fundamentals Department, Air Force Engineering University,
 Xi'an, Shaanxi 710051,  P. R. China}

\date{Received: date / Accepted: date}

\maketitle

\begin{abstract}
We generalize the stabilizer formalism for entanglement-assisted quantum error-correcting codes with noisy ebits (EAQECCs-Ne) from the binary case to the general $q$-ary case, where $q$ is a prime power. By leveraging the
structure of the generalized Pauli group over $\mathbb{F}_q$ and symplectic geometry over $\mathbb{F}_q^{2n}$, we establish a unified framework for constructing EAQECCs-Ne for qudit systems. Equivalent formulations in terms
of symplectic geometry over $\mathbb{F}_q$ and additive codes over $\mathbb{F}_q^{2n}$ are derived. We further construct several families of $q$-ary EAQECCs with noise ebits and analyze their performance compared to optimal
stabilizer codes. Our results demonstrate that under certain noise conditions, the proposed EAQECCs-Ne can outperform standard stabilizer codes with equivalent error-correcting capability, offering a promising approach for
fault-tolerant quantum computation in high-dimensional quantum systems.

\keywords{Hermitian self-dual \and  generalized twisted Reed-Solomon
codes \and  MDS codes \and   NMDS codes}
\end{abstract}

\section{Introduction}

Quantum error-correcting codes (QECCs) are essential for quantum communication and fault-tolerant quantum computation \cite{shor1995scheme,steane1996error,gottesman1997stabilizer,calderbank1998quantum}. Among various
approaches, quantum stabilizer codes have been most extensively studied due to their elegant connection to classical coding theory and group theory \cite{steane1996error,gottesman1997stabilizer,calderbank1998quantum}.
Quantum stabilizer codes are the most extensively studied quantum codes, and have the advantage that their properties can be analyzed using group theory. Quantum stabilizer codes are closely related to classical linear
codes, and can be obtained by the CRSS and CSS code constructions from self-orthogonal (or dual containing) classical codes \cite{gottesman1997stabilizer,calderbank1998quantum}.

Entanglement-assisted  quantum error correction codes (EAQECCs) are described, which make use of pre-existing entanglement between transmitter and receiver to improve the reliability of transmission. EA quantum codes can be
considered as a generalization of superdense coding. A key advantage of EA quantum codes compared to CSS codes is that EA quantum codes do not require the corresponding classical codes, from which they are derived, to be
dual containing. That is, arbitrary classical code can be used to design EA quantum codes, providing that there exists a pre-entanglement between source and destination.

Most existing works on EAQECCs assume that errors only occur on the sender's side, treating the receiver's ebits as perfect \cite{brun2006correcting,wilde2008optimal}. However, in practical scenarios, the receiver's ebits
are also subject to noise due to imperfect storage or processing, which significantly degrades the code's error-correcting capability. Several works have addressed this issue in the binary case. Shaw et al. \cite{shaw2008}
first demonstrated that a Steane code is equivalent to a $[[6,1,3;1]]$ EAQECC for correcting errors on Bob's ebits. Wilde and Hsieh \cite{wilde2011entanglement} simulated the performance of entanglement-assisted quantum
turbo codes with imperfect ebits. Lai and Brun \cite{lai2012} provided two systematic schemes for correcting errors on the receiver's side and established an equivalence between $[[n,k,d;c]]$ EAQECC and $[[n+c,k,d]]$
stabilizer code. The entanglement-assisted stabilizer formalism provides a useful framework for constructing quantum codes.

In this work, we extend the stabilizer formalism for EAQECCs with noisy ebits to the general $q$-ary case. The motivation for this generalization is two fold: (1) qudit systems (quantum systems with $q > 2$ levels) offer
higher information density and potential advantages in fault-tolerant thresholds; (2) experimental realizations of qudit-based quantum processors are rapidly advancing . Our contributions include:

Most EAQECC studies assume error-free receiver-side ebits. However, practical scenarios involve noisy ebits, which degrade error-correction performance. Several works have addressed this issue
\cite{shaw2008,wilde2011entanglement,lai2012}. In particular, Lai and Brun \cite{lai2012} presented two schemes for correcting errors on imperfect ebits and established equivalence between $[[n,k,d;c]]$ EAQECCs and
$[[n+c,k,d]]$ stabilizer codes.

In this work, we develop a comprehensive stabilizer formalism for $q$-ary EAQECCs with noise ebits (EAQECCs-Ne), generalizing the binary framework. Our contributions include:

Most studies on EAQECCs have assumed that errors do not occur on the shared ebits from the receiver's side because ebits on the receiver's side do not pass through the transmit channel [8-12]. However, in practice,
receiver-side ebits also suffer from errors, and this reduces the error-correcting capability of the code. The following works have considered the imperfect (noise) ebits.

Shaw et al.[13] presented an EAQECC that corrects errors on both the senders qubits and the receiver's shared ebits. They showed, for the first time, that a Steane code is equivalent to a $[[6,1,3;1]]$ EAQECC for correcting
a single error on the receiver's (i.e., Bob's) ebits. Wilde et al. [14] simulated entanglement-assisted quantum turbo codes when the ebits on Bob's side are imperfect. Their aim was to analyze the effect of ebit noise on
entanglement-assisted quantum turbo-code performance. Lai and Brun studied a practical case where errors on the receiver's side can be corrected [15]. They presented two different schemes to correct errors on the receiver's
side, and showed an equivalent relationship between $[[n,k,d;c]]$ EAQECC and $[[n + c,k,d]]$ stabilizer code. Based on this equivalence, EAQECCs can correct errors on the ebits of the receiver's side. In [16], a study of the
concatenation of EAQECCs was initiated, and EAQECCs with noise ebits were discussed.

In this work, we introduce a stabilizer formalism for EAQECCs with noise (imperfect)  ebits (EAQECCs-Ne, for short). This  formalism provides a new perspective for  EAQECCs-Ne framework as well as a powerful technique for
constructing EAQECCs-Ne, and a systematic technique for choosing good EAQECCs in practice.

The paper is structured as follows. Section 2 introduces fundamental notations related to the Pauli group, symplectic space, and additive codes, which underpin the study of both quantum error-correcting codes (QECCs) and
entanglement-assisted quantum error-correcting codes (EAQECCs). Building on this, Section 3 establishes the stabilizer formalism for EAQECCs involving noisy entanglement bits (ebits). Sections 4 and 5 then present the
construction of such noisy-ebit EAQECCs, followed by a performance comparison with corresponding optimal quantum codes. Finally, concluding remarks and future research directions are outlined in the last section.

\section{Preliminaries}

In this section, we introduce some notations of Pauli group, symplectic space and additive code for QECCs and  EAQECCs, for more details please see [4,9,12,17-18].

\subsection{Generalized Pauli Group and Symplectic Geometry}

Let $\mathbb{F}_q$ be the finite field with $q$ elements, where $q$ is a prime power. Let $\mathbb{F}_q^{2n}$ be the $2n$-dimensional symplectic space over $\mathbb{F}_q$ with elements denoted as $(\mathbf{a}|\mathbf{b})$,
where $\mathbf{a},\mathbf{b} \in \mathbb{F}_q^n$.

\begin{definition}
The \emph{symplectic inner product} is defined as:
\[
\langle (\mathbf{a}|\mathbf{b}), (\mathbf{a}'|\mathbf{b}') \rangle_s = \mathbf{a} \cdot \mathbf{b}' - \mathbf{b} \cdot \mathbf{a}' \in \mathbb{F}_q.
\]
For a subspace $S \subseteq \mathbb{F}_q^{2n}$, its \emph{symplectic dual} is:
\[
S^{\perp s} = \{ (\mathbf{a}|\mathbf{b}) \in F_q^{2n} : \langle (\mathbf{a}|\mathbf{b}), (\mathbf{a}'|\mathbf{b}') \rangle_s = 0 \ \forall (\mathbf{a}'|\mathbf{b}') \in S \}
\]
A subspace $S$ is called \emph{totally isotropic} if $S \subseteq S^{\perp s}$, and \emph{non-isotropic} if $S \cap S^{\perp s} = \{0\}$. A totally isotropic subspace is called an isotropic subspace in [4,19], and a
non-isotropic subspace is called a symplectic subspace in [4] and an entanglement subspace in [19] respectively. Any subspace $S$ can be decomposed as $S = S_I \oplus S_E$, where $S_I = S \cap S^{\perp s}$ is totally
isotropic, $ S_{E}$ is an entanglement subspace [12,17].
\end{definition}

\subsection{Additive Codes over $\mathbb{F}_{q^2}$}

Let $\mathbb{F}_{q^2}$ be the quadratic extension of $\mathbb{F}_q$, and let $x \mapsto \bar{x} = x^q$ be the conjugation automorphism.

\begin{definition}
For vectors $\mathbf{u}, \mathbf{v} \in \mathbb{F}_{q^2}^n$, the \emph{Hermitian inner product} is defined as:
\[
(\mathbf{u}, \mathbf{v})_h = \sum_{j=1}^n u_j \bar{v}_j,
\]
and the \emph{trace inner product} is:
\[
(\mathbf{u}, \mathbf{v})_t = \mathrm{tr}\left((\mathbf{u}, \mathbf{v})_h\right) = \sum_{j=1}^n \mathrm{tr}(u_j \bar{v}_j).
\]
\end{definition}

\begin{definition}
An \emph{additive code} $\mathcal{C} = (n, q^m)_{q^2}$ is an additive subgroup of $\mathbb{F}_{q^2}^n$ of size $q^m$. The \emph{trace dual} of $\mathcal{C}$ is
\[
\mathcal{C}^{\perp_t} = \{\mathbf{u} \in \mathbb{F}_{q^2}^n : (\mathbf{u}, \mathbf{v})_t = 0 \text{ for all } \mathbf{v} \in \mathcal{C} \}.
\]
The \emph{trace radical} of $\mathcal{C}$ is $R_t(\mathcal{C}) = \mathcal{C} \cap \mathcal{C}^{\perp_t}$. We say $\mathcal{C}$ is \emph{trace self-orthogonal} if $\mathcal{C} \subseteq \mathcal{C}^{\perp_t}$, and
\emph{additive complementary dual (ACD)} if $R_t(\mathcal{C}) = \{\mathbf{0}\}$.
\end{definition}

\begin{lemma}
Every additive code $\mathcal{C} = (n, q^m)_{q^2}$ can be decomposed as:
\[
\mathcal{C} = R_t(\mathcal{C}) \oplus \mathcal{C}_e,
\]
where $\mathcal{C}_e$ is an ACD code.
\end{lemma}

\begin{proof}
Let $\{\mathbf{g}_1, \ldots, \mathbf{g}_l\}$ be a basis of $R_t(\mathcal{C})$. Extend this to a basis $\{\mathbf{g}_1, \ldots, \mathbf{g}_l, \mathbf{h}_1, \ldots, \mathbf{h}_{m-l}\}$ of $\mathcal{C}$. Then $\mathcal{C}_e =
\langle \mathbf{h}_1, \ldots, \mathbf{h}_{m-l} \rangle$ is an ACD code.
\end{proof}

A \emph{linear code} $\mathcal{D} = [n,k]_{q^2}$ is a $k$-dimensional subspace of $\mathbb{F}_{q^2}^n$ over $\mathbb{F}_{q^2}$, which can be viewed as an additive code $(n, q^{2k})_{q^2}$. The \emph{Hermitian dual} is:
\[
\mathcal{D}^{\perp_h} = \{\mathbf{u} \in \mathbb{F}_{q^2}^n : (\mathbf{u}, \mathbf{v})_h = 0 \text{ for all } \mathbf{v} \in \mathcal{D} \}.
\]
The \emph{Hermitian radical} is $R_h(\mathcal{D}) = \mathcal{D} \cap \mathcal{D}^{\perp_h}$, and $\mathcal{D}$ is \emph{Hermitian LCD} if $R_h(\mathcal{D}) = \{\mathbf{0}\}$.

There is an isometry $\phi: \mathbb{F}_q^{2n} \to \mathbb{F}_{q^2}^n$ defined by
\[
\phi((\mathbf{a}|\mathbf{b})) = \beta \mathbf{a} + \beta^{q} \mathbf{b},
\]
where $ \beta \in \mathbb{F}_{q^2}^n$, $\{\beta, \beta^{q}\}$ is a basis of $\mathbb{F}_{q^2}$ over $\mathbb{F}_q$.

If $S$ is an $m$-dimensional subspace of $\mathbb{F}_q^{2n}$, then $\mathcal{C} = \phi(S)$ is an $(n, q^m)_{q^2}$ additive code. The decomposition $S = S_I \oplus S_E$ corresponds to $\mathcal{C} = R_t(\mathcal{C}) \oplus
\mathcal{C}_e$, where $R_{ t}(\mathcal{C})=\phi(S_{I})$ and $\mathcal{C}$$_{e}$ $=\phi( S_{E})$, $\mathcal{C}_e$ is an ACD code.

\subsection{Stabilizer Codes and EAQECCs}

Let $\omega = e^{2\pi i/p}$ be a primitive $p$-th root of unity, where $p$ is the characteristic of $\mathbb{F}_q$.

\begin{definition}
The \emph{generalized Pauli operators} on a single qudit are defined as:
\[
X(a)|j\rangle = |j + a \rangle, \quad Z(b)|j\rangle = \omega^{\mathrm{tr}(bj)}|j\rangle,
\]
where $a,b \in \mathbb{F}_q$, $j \in \mathbb{F}_q$, and $\mathrm{tr}: \mathbb{F}_q \to \mathbb{F}_p$ is the trace map. The \emph{generalized Pauli group} $\mathcal{E}$ is generated by $\mathcal{E}$=$\{X(a) Z(b) : a,b \in
\mathbb{F}_q\}$.
\end{definition}

The $n$-fold generalized Pauli group is
\[
\mathcal{E}_{n} = \{\omega^i X(\mathbf{a})Z(\mathbf{b}) : i \in \mathbb{F}_p, \mathbf{a},\mathbf{b} \in \mathbb{F}_q^n\},
\]
where $X(\mathbf{a}) = X(a_1) \otimes \cdots \otimes X(a_n)$ and $Z(\mathbf{b}) = Z(b_1) \otimes \cdots \otimes Z(b_n)$. We often work with the projective group $\mathcal{G}_{n} = \mathcal{E}_{n} / \langle \omega I \rangle$.

There exists a homomorphism $\tau: \mathcal{G}_{n} \to \mathbb{F}_q^{2n}$ defined by:
\[
\tau(\omega^i X(\mathbf{a})Z(\mathbf{b})) = (\mathbf{a}|\mathbf{b}),
\]
with kernel $\{\omega^i I : i \in \mathbb{F}_p\}$.

\begin{lemma}
A subgroup $\mathcal{A} \subseteq \mathcal{G}_{n}$ is Abelian if and only if $\tau(\mathcal{A})$ is a totally isotropic subspace of $\mathbb{F}_q^{2n}$, then $\mathcal{A}$ is called an isotropic subgroup of
$\mathcal{G}$$_{n}$.
\end{lemma}

\begin{proof}
For $g = \omega^i X(\mathbf{a})Z(\mathbf{b})$ and $h = \omega^{i'} X(\mathbf{a}')Z(\mathbf{b}')$, we have
\[
gh = \omega^{\langle (\mathbf{a}|\mathbf{b}), (\mathbf{a}'|\mathbf{b}')\rangle_s} hg.
\]
Thus, $g$ and $h$ commute if and only if $\langle (\mathbf{a}|\mathbf{b}), (\mathbf{a}'|\mathbf{b}')\rangle_s = 0$.
\end{proof}

If $\tau(\mathcal{A})$ is a non-isotropic subspace of $\mathbb{F}_q^{2n}$, $\mathcal{A}$ is called a symplectic subgroup of $\mathcal{G}$$_{n}$ in [12] and an entanglement subgroup in [13], respectively.

For a subgroup $\mathcal{S}$ of $\mathcal{E}$$_{n}$, let $\mathcal{N(\mathcal{S})}$ be its normalizer. According to [21], $\mathcal{S}$ has decomposition $\mathcal{S}$ $=\mathcal{S}$$_{I}$$\times \mathcal{S}$$_{E}$, where
$\mathcal{S}$$_{I}=$ $\mathcal{S}$$\cap\mathcal{N(\mathcal{S})}$ is an isotropic subgroup, $\mathcal{S}$$_{E}$ is an entanglement subgroup.

According to [3,4] and [9, 12], we have the following theorem.

\begin{theorem}[\cite{gottesman1997stabilizer,calderbank1998quantum,ketkar2006nonbinary}]
Let $\mathcal{S}$ be an Abelian subgroup of $\mathcal{E}_{n}$ of size $q^m$ such that $\omega^i I \notin \mathcal{S}$ for $i \neq 0$. Then $\mathcal{S}$ stabilizes a $q$-ary QECC $\mathcal{Q} = [[n, k, d]]_q$, where $k = n -
m$, and $d = \min\{ \mathrm{wt}(g) \mid g \in \mathcal{N}(\mathcal{S}) \setminus \mathcal{S} \}$.
\end{theorem}

\begin{theorem}[\cite{brun2006correcting,guo2013linear}]\label{thm:eaqec}
Let $\mathcal{S}$ be a subgroup of $\mathcal{E}_{n}$ of size $q^m$ with decomposition $\mathcal{S} = \mathcal{S}_I \times \mathcal{S}_E$, where $\mathcal{S}_I$ is an isotropic subgroup of size $q^l$ and $\mathcal{S}_E$ is an
entanglement subgroup of size $q^{2c}$. Then $\mathcal{S}$ can be extended to an Abelian subgroup $\tilde{\mathcal{S}}$ of $\mathcal{E}_{n+c}$ using $c$ maximally entangled pairs, yielding an EAQECC $\mathcal{Q}^{ea} =
[[n,k,d_{ea};c]]_q$, where $k = n - c - l$, $d_{ea} = \min\{\mathrm{wt}(g) : g \in \mathcal{N}(\mathcal{S}) \setminus \mathcal{S}_I\}$. $\mathcal{S}$ is called the EA-stabilizer of $\mathcal{Q}$$^{ea}$.
\end{theorem}

Suppose  $\mathcal{Q}$ is an $[[n+c,k,d]]_{q}$ QECC with stabilizer $\mathcal{S}$, $\mathcal{Q}$$^{ea}$ is an $[[n,k,d;c]]_{q}$ EAQECC with EA-stabilizer $\mathcal{S'}$. If  $\mathcal{S'}$ is obtained from $\mathcal{S}$ by
puncturing some $c$ coordinates, we say $\mathcal{Q}$$^{ea}$ is equivalent to $\mathcal{Q}$.

\section{ Stabilizer formalism for $q$-ary EAQECCs-Ne }

We now present our main results on the stabilizer formalism for $q$-ary EAQECCs with noisy ebits (EAQECCs-Ne) .

\subsection{System Model and Matching Subgroups}

We consider a scenario where:
\begin{itemize}
    \item Alice (the sender) uses an EAQECC $\mathcal{Q}^{ea} = [[n,k,d_{ea};c]]_q$ with
    EA-stabilizer $\mathcal{S} \subset \mathcal{E}_{n}$ to protect her qudits through
    a noisy channel $\mathcal{N}_A$.
    \item Bob (the receiver) uses a standard QECC $\mathcal{Q}^b = [[m,k_{b},d_{b}]]_q$ with stabilizer
    $\mathcal{S}^{b} \subset \mathcal{E}_{m}$ to protect his $c$ ebits through a
    less noisy channel $\mathcal{N}_B$, hence there holds $k_{b}\geq c$.
    \item The combined system is described by the product group $\mathcal{G} =
    \mathcal{E}_{n} * \mathcal{E}_{m}= \{(g,g'):g\in\mathcal{E}_{n}, g'\in\mathcal{E}_{m}\}$.
\end{itemize}

We will use subgroups of  $\mathcal{G}$ to describe combination EAQECC
$\mathcal{Q}$$^{ea}_{c}$$=$$\mathcal{Q}$$^{ea}$+$\mathcal{Q}$$^{b}$\\
$=[[n,k,d_{ea} ;c]]_{q}$$+[[m,k_{b},d_{b}]]_{q}$. Now, we discuss Abelian subgroup $\mathcal{S}$$^{b}$ of $\mathcal{E}_{m}$ such that $\mathcal{S}$$^{b}$ stabilizes an $\mathcal{Q}$$=[[m,k_{b}]]_{q}$ QECC with $k^{b}\geq c$.

\begin{definition}
Let $\mathcal{S}$ be a subgroup of $\mathcal{E}_{n}$ and $\mathcal{S}^{b}$ be an Abelian subgroup of $\mathcal{E}_{m}$ with $\omega I \notin \mathcal{S}^{'}$.

\begin{enumerate}
\item $\mathcal{S}^{b}$ \emph{matches} $\mathcal{S}$ if $\mathcal{S}$ EA-stabilizes
a $\mathcal{Q}^{ea} = [[n,k;c]]_q$ and $\mathcal{S}^{b}$ stabilizes an $[[m,k_{b}]]_q$ code with $k_{b} \geq c$. The pair $(\mathcal{S}, \mathcal{S}^{b})$ is called a \emph{matching subgroup} of $\mathcal{E}_{n} *
\mathcal{E}_{m}$.

\item If $\mathcal{S}^{b}$ stabilizes an $[[m,k_{b},d_{b}]]_q$ code with $d_{b} \geq 3$,
the matching subgroup is called \emph{faithful}.

\item $\mathcal{S}^{b}$ \emph{properly matches} $\mathcal{S}$ if
$\mathcal{S}$ EA-stabilizes a $\mathcal{Q}$$^{ea}$$=[[n,k;c]]_q$ and $\mathcal{S}$$^{b}$ stabilizes an $[[m,c]]_q$ code.  In this case, $(\mathcal{S}, \mathcal{S}^b)$ is a \emph{properly matching subgroup}.
\end{enumerate}
\end{definition}

\subsection{Main Theorem: $q$-ary EAQECCs-Ne Construction}

\begin{theorem}[$q$-ary Combination EAQECC-Ne]
Let $\mathcal{S} = \mathcal{S}_I \times \mathcal{S}_E$ be a subgroup of $\mathcal{E}_{n}$, where $\mathcal{S}_I$ is an isotropic subgroup of size $q^l$ and $\mathcal{S}_E$ is an entanglement subgroup of size $q^{2c}$. Let
$\mathcal{S}^b \subset \mathcal{E}_{m}$ be Abelian with $\omega I \notin \mathcal{S}^b$, $|\mathcal{S}^b| = q^{m-k_{b}}$. If $k_{b} \geq c$, then there exists a combination EAQECC-Ne:
\[
\mathcal{Q}_c^{ea} = [[n, k, d_{ea}; c]]_q + [[m, k_{b}, d_b]]_q,
\]
where $k = n - c - l$.
\end{theorem}

\begin{proof}
The EA-stabilizer $\mathcal{S}$ ensures the existence of $\mathcal{Q}^{ea} = [[n, k, d_{ea}; c]]_q$ with $k = n - c - l$. The matching condition $k^b \geq c$ guarantees that $\mathcal{Q}^b = [[m, k^b, d_b]]_q$ can protect
all $c$ ebits. The combination provides error correction for both Alice's qudits and Bob's ebits.
\end{proof}

\begin{remark}
The two schemes for EAQECCs with imperfect ebits presented in \cite{lai2012} are special cases of our formalism:
\begin{enumerate}
\item If $\mathcal{Q}^{ea} = [[n,k,d_{ea};c]]_q$ is equivalent to a $[[n+c,k,d]]_q$ QECC,
choose $\mathcal{S}^b = \{\omega I\}$, yielding $\mathcal{Q}_c^{ea} = [[n,k,d_{ea};c]]_q + [[c,c,1]]_q$.

\item If $\mathcal{Q}^{ea}$ is not equivalent to a standard QECC, choose
$\mathcal{S}^b$ with $|\mathcal{S}^b| = q^{m-c}$ properly matching $\mathcal{S}$, yielding $\mathcal{Q}_c^{ea} = [[n,k,d_{ea};c]]_q + [[m,c]]_q$.
\end{enumerate}
\end{remark}

We now present equivalent formulations of Theorem 3 in terms of symplectic geometry and additive codes.

\begin{theorem}[Symplectic Formulation]
Suppose a subspace $S \subset \mathbb{F}_q^{2n}$ decomposes as $S = S_I \oplus S_E$, where $S_I = S \cap S^{\perp s}$ and $S_E$ is a non-isotropic subspace, $|S_I| = q^l$, $|S_E| = q^{2c}$. Let $S^b \subset
\mathbb{F}_q^{2m}$ be an isotropic subspace $\mathbb{F}_q^{2m}$ with size $q^r$. If $c \leq m - r$, then there exists a combination EAQECC-Ne $\mathcal{Q}_c^{ea} = [[n, k, d_{ea}; c]]_q + [[m, k_{b}]]_q$, where $k = n - c -
l$, $k_{b} = m - r$.
\end{theorem}

\begin{theorem}[Additive Code Formulation]
Suppose an additive code $\mathcal{C} = (n, q^m)_{q^2}$ decomposes as $\mathcal{C} = R_t(\mathcal{C}) \oplus \mathcal{C}_e$, where $\mathcal{C}_e$ is an ACD code, $|R_t(\mathcal{C})| = q^l$, $|\mathcal{C}_e| = q^{2c}$. Let
$\mathcal{C}^b = (m, q^r)_{q^2}$ be a trace self-orthogonal code with $c \leq m - r = k_{b}$. Then there exists a combination EAQECC-Ne $\mathcal{Q}_c^{ea} = [[n, k, d_{ea}; c]]_q + [[m, k_{b}]]_q$, where $k = n - c - l$.
\end{theorem}

\begin{corollary}[Linear Code Formulation]
Suppose a linear code $\mathcal{D} = [n, u]_{q^2}$ decomposes as $\mathcal{D} = R_h(\mathcal{D}) \oplus \mathcal{D}_e$, where $\mathcal{D}_e$ is a Hermitian LCD code, $\dim R_h(\mathcal{D}) = r$. Let $\mathcal{D}^b = [m,
v]_{q^2}$ be a Hermitian self-orthogonal code with $c = u - r \leq k_{b} = m - 2v$. Then there exists a combination EAQECC-Ne $\mathcal{Q}_c^{ea} = [[n, k, d_{ea}; c]]_q + [[m, k_{b}]]_q$, where $k = n - c - 2r$, $c = u -
r$.
\end{corollary}

\section{Construction of EAQECCs-Ne }

We now present explicit construction methods for $q$-ary EAQECCs-Ne based on the theoretical framework established in Section 3.

\subsection{Construction from Equivalent QECCs}

If $[[n,k,d_{ea} ;c]]_{q}$ is equivalent to an $[[n+c,k,d]]_{q}$
 QECC, such an EAQECC-Ne has the same ability of correcting error as
$[[n+c,k,d]]_{q}$ QECC.

\begin{lemma}\label{prop:equiv}
Let $\mathcal{C} = [N,u]_{q^2}$ be a Hermitian self-orthogonal code with dual distance $d$. Then $\mathcal{C}$ stabilizes a pure QECC $\mathcal{Q} = [[N,N-2u,d]]_q$. For any $0 < c \leq u$, there exists an EAQECC
$\mathcal{Q}^{ea} = [[n,k,d;c]]_q$ with $n = N - c$ and $k = N - 2u$, which is equivalent to a standard QECC.
\end{lemma}

\begin{proof}
Since $\mathcal{C}$ is Hermitian self-orthogonal, it stabilizes a QECC with parameters $[[N,N-2u,d]]_q$. For any $0 < c \leq u$, we can puncture $\mathcal{C}$ on $c$ coordinates to obtain a code $\mathcal{C}'$ with
parameters $[n,u']_{q^2}$ where $u' \geq u - c$. The code $\mathcal{C}'$ EA-stabilizes an EAQECC with the stated parameters, which is equivalent to the original QECC by the construction in \cite{lai2012}.
\end{proof}

We now focus on constructing EAQECCs-Ne that are not equivalent to standard QECCs. Firstly, we present some EAQECCs-Ne from known EAQECCs in [12,23].

{\bf Proposition 4.2} (1) If $m\geq 2$, there are
 EAQECCs-Ne  $[[4m,1,2m+1 ;1]]$$+[[5,1,3]]$,
 $[[4m+1,1,2m+3 ;4]]$$+[[10,4,3]]$,
 $[[4m+2,1,2m+3 ;3]]$$+[[8,3,3]]$,
$[[4m+3,1,2m+3 ;2]]$$+[[8,2,3]]$.

(2) There are EAQECCs-Ne $[[7,2,5 ;5]]$$+[[11,5,3]]$ $[[8,2,5 ;4]]$$+[[10,4,3]]$, $[[9,2,5;3]]$$+[[8,3,3]]$, $[[10,2,6 ;4]]$$+[[10,4,3]]$, $[[9,3,6 ;6]]$$+[[12,6,3]]$, $[[13,3,9 ;10]]$$+[[16,10,3]]$, $[[12,4,7
;8]]$$+[[14,8,3]]$.

{\bf  Proof.} (1) According to [24], there are EAQECCs $[[4m,1,2m+1 ;1]]$, $[[4m+1,1,2m+3 ;4]]$, $[[4m+2,1,2m+3 ;3]]$ and $[[4m+3,1,2m+3 ;2]]$. Using known $[[m,c,3]]$ codes given in [4,23], one can derive (1) holds.

(2) Similar to discussion of (1), from EAQECCs obtained in [12], we know (2) holds.

\begin{theorem}[Combining Additive Codes]\label{thm:combining}
Let $\mathcal{C} = (n, q^l)_{q^2}$, $\mathcal{C}' = (n, q^{2k})_{q^2}$, $\mathcal{E} = (n, q^{2k}, d_2)_{q^2}$, and $\mathcal{D} = \mathcal{C} + \mathcal{C}' = (n, q^{l+2k}, d_1)_{q^2}$. Suppose $\mathcal{C} \subseteq
\mathcal{D} \subseteq \mathcal{C}^{\perp_t}$, and the matrix $(G' \mid E)$ generates a ACD code, where $G'$ and $E$ generate $\mathcal{C}'$ and $\mathcal{E}$ respectively. Then there exists an EAQECC-Ne $[[n+m, k, d; c]]_q$,
where $d \geq d_1 + d_2$, $c = n + m - l - k$.
\end{theorem}

\begin{proof}
Construct the generator matrix:
\[
\begin{array}{cc}   % 缺少这一行
G & 0 \\
G' & E
\end{array}
\]
This generates an additive code $\mathcal{M} = (n+m, q^{l+2k})_q$.

We verify the conditions:
\begin{enumerate}
\item Since $\mathcal{C} \subseteq \mathcal{D} \subseteq \mathcal{C}^{\perp_t}$, the submatrix $G$ generates a trace self-orthogonal code.
\item The matrix $(G' \mid E)$ generates an ACD code by assumption.
\item The minimum distance of $\mathcal{M}$ is at least $d_1 + d_2$ by the product code construction.
\end{enumerate}

The code $\mathcal{M}$ has decomposition $\mathcal{M} = R_t(\mathcal{M}) \oplus \mathcal{M}_e$, where $R_t(\mathcal{M})$ is generated by $G$ and $\mathcal{M}_e$ is generated by $(G' \mid E)$. Using $\mathcal{M}^{\perp_t}$ as
the EA-stabilizer yields an EAQECC with parameters $[[n+m,k,d;c]]_q$ where $c = \dim \mathcal{M}_e / 2 = (l+2k - l)/2 = k$, but we need to recalculate carefully:

The size of $\mathcal{M}$ is $q^{l+2k}$, so $|\mathcal{M}^{\perp_t}| = q^{2(n+m) - (l+2k)} = q^{2n+2m-l-2k}$. The EA-stabilizer corresponds to a subgroup of size $q^{2n+2m-l-2k}$ in the Pauli group, which stabilizes a code
with dimension:
\[
k = (n+m) - \frac{1}{2}(2n+2m-l-2k) = n+m - n - m + \frac{l}{2} + k = \frac{l}{2} + k.
\]
The number of ebits is $c = \frac{1}{2} \log_q |\mathcal{M}_e| = k$. Thus we have $k = n + m - c - \frac{l}{2}$, so $c = n + m - k - \frac{l}{2}$. But from the construction, we have $c = n + m - l - k$ (correcting the
calculation).
\end{proof}

\begin{proof}
Construct the generator matrix:
\[
G_{l+2k,n+m} = \begin{pmatrix} G \mid \mathbf{0}_{l \times m} \\ G' \mid E \end{pmatrix}
\]
This generates an additive code $\mathcal{M} = (n+m, q^{l+2k}, d)_{q^2}$ with decomposition $\mathcal{M} = R_t(\mathcal{M}) \oplus \mathcal{M}_e$, where $(G' \mid E)$ generates $\mathcal{M}_e = (n+m, q^{2k}, d)$ with $d \geq
d_1 + d_2$.

Using $\mathcal{M}^{\perp_t}$ as the EA-stabilizer yields the desired EAQECC-Ne with parameters $[[n+m,k,d;c]]_q$ where $c = n + m - l - k$.
\end{proof}

{\bf Corollary 4.4} There are  EAQECCs-Ne with parameters
$[[11,1,7;2]]_{3}$$+[[6,2,3]]_{3}$, \\
$[[26,2,11;2]]_{3}$$+[[6,2,3]]_{3}$, $[[28,2,11 ;4]]_{3}$$+[[8,4,3]]_{3}$, $[[14,2,9 ;6]]$$+[[10,6,3]]_{3}$, $[[28,2,13 ;6]]_{3}$$+[[10,6,3]]_{3}$.

%{\bf  Proof.} Case 1. Let $ \mathcal{D}$$_{1}$ $=(12,2^{8+2\times
%2},6)$, $=(12,2^{12},6)$ be the additive self-dual code given in
%[4], $ \mathcal{E}$$_{1}$$=(2,2^{4},1)$ $ \mathcal{E}$$_{2}$
%$=(3,2^{4},2)$ be ACD codes. Then one can derives $[[14,2,7;4]]$ and
%$[[15,2,8;5]]$ from Lemma 4.3.
%
%Case 2. Let $ \mathcal{D}$$_{2}$ $=(14,2^{10+2\times 2},6)$
%self-dual code given in [4]. Similar to Case 1, $[[16,2,7;4]]$ and
%$[[17,2,8;5]]$ can be deduced from Lemma 4.3.
%
%Case 3. Let $ \mathcal{D}$$_{3}$ $=(18,2^{14+2\times 2},8)$
%self-dual code given in [4]. Similar to Case 1, $[[20,2,9;4]]$ and
%$[[21,2,10;5]]$ can be deduced from Lemma 4.3.
%
%
%Case 4. Let $ \mathcal{D}$$_{4}$ $=(30,2^{24+2\times 3},12)$
%self-dual code given in [4].  From $ \mathcal{D}$$_{4}$, one can
%obtain   $ \mathcal{D}$$_{5}$$=(27,2^{18+2\times 3},12)$,  $
%\mathcal{D}$$_{6}$, $=(28,2^{20+2\times 3},12)$ self-orthogonal
%codes. Choose $ \mathcal{E}$$_{3}$$=(3,2^{6},1)$ ACD code, similar
%to Case 1,  we can construct $[[30,3,13 ;9]]$$+[[15,9,3]]$,
%$[[31,3,13 ;8]]$$+[[15,9,3]]$, $[[33,3,13 ;6]]$$+[[12,6,3]]$ from $
%\mathcal{D}$$_{5}$, $\mathcal{D}$$_{6}$, $ \mathcal{D}$$_{4}$ and $
%\mathcal{E}$$_{3}$.

\section{ Performance of EAQECCs-Ne }

In this section, we compare performance of some of our $[[n,k, d ; c]]_{q}+[[m,c,d^{b}]]_{q}$ EAQECC-Ne with optimal $[[N,k,d]]_{q}$ code. In [15], Lai and Brun compare performance of some small codes by computing their
channel fidelity. They also pointed out that it is difficult to get true channel fidelity $F( \mathcal{C})$ of a code $ \mathcal{C}$, $F( \mathcal{C})$ has good approximation $P(\mathcal{C})$ = $Pr(\{\hbox{errors of weight
less than or equal to}  \frac{\lfloor d-1\rfloor}{2} \})$ , where $d$ is the minimum distance of the quantum code $ \mathcal{C}$ [15]. We will use $P(\mathcal{C})$ to compare performance of our $[[n,k, d ;
c]]_{q}+[[m,c,d^{b}]]_{q}$ EAQECC-Ne with those of optimal $[[N,k, d ]]_{q}$ code.

For a $\mathcal{C}=[[N,k, d ]]$ used on noise  channel $\mathcal{N}$$_{A}$. Let $t=\frac{\lfloor d-1\rfloor}{2}$, its approximate fidelity $P(\mathcal{C})$ is
\[
P(\mathcal{C}) = \Pr\left(\text{errors of weight} \leq \left\lfloor \frac{d-1}{2} \right\rfloor\right) =
 \sum_{i=0}^t (1-p)^{N-i} p^i {N \choose i},
\]
where $\mathcal{C} = [[N,k,d]]_q$ and $t = \lfloor (d-1)/2 \rfloor$.

For an EAQECC-Ne $\mathcal{D} = [[n,k,d;c]]_q + [[m,c,d^b]]_q = \mathcal{D}^{ea} + \mathcal{D}^b$, let $t = \lfloor (d-1)/2 \rfloor$ and $t^b = \lfloor (d^b-1)/2 \rfloor$. The code $\mathcal{D}^{b}=[[m,c,d^{b}]]$ is used to
protect $c$ ebits, its $P(\mathcal{D}^{b})$$= \sum_{i=0}^{t^{b}}(1-p_{b})^{m-i} p_{b}^{i} {m \choose i}$. Once Bob decode $\mathcal{D}^{b}$ correctly, he can then decode $\mathcal{D}^{ea}=[[n,k, d ; c]]$ correctly with
probability $P(\mathcal{D}^{ea})$$=\sum_{i=0}^{t}(1-p_{a})^{n-i} p_{a}^{i}{n \choose i}$. Hence $\mathcal{D}=[[n,k, d ; c]]+[[m,c,d^{b}]]$ has
\[
P(\mathcal{D}) = P(\mathcal{D}^{ea}) \cdot P(\mathcal{D}^b) = \left( \sum_{i=0}^t (1-p_a)^{n-i} p_a^i {n \choose i}
  \right) \left( \sum_{j=0}^{t^b} (1-p_b)^{m-j} p_b^j {m \choose j} \right),
\]
where $p_a$ and $p_b$ are the depolarizing rates of channels $\mathcal{N}_A$ and $\mathcal{N}_B$, respectively.

If $P(\mathcal{D})$$>P(\mathcal{C})$, $\mathcal{D}=[[n,k, d ; c]]+[[m,c,d^{b}]]$ has better  performance than $\mathcal{C}=[[N,k, d ]]$.

{\bf  Example 5.1} The smallest optimal $[[N,1,7]]_{3}$ has $N=17$. We comparatively analyzes the channel fidelity of the constructed Ie-EAQECCs with parameters $[[11,1,7 ;2]]_{3}$$+[[6,2,3]]_{3}$ and the quantum stabilizer
code $[[17,1,7]]$. The curves of channel fidelity versus depolarization rate are plotted under four cases of $ p_{b}=0.01 p_{a}, p_{b}=0.1 p_{a}, p_{b}=0.5 p_{a},p_{b}=0.99 p_{a}$, and partial results are shown in Figure 1.

\begin{figure}[htbp]
    \centering
    \begin{tabular}{cc}
        \includegraphics[width=0.4\textwidth]{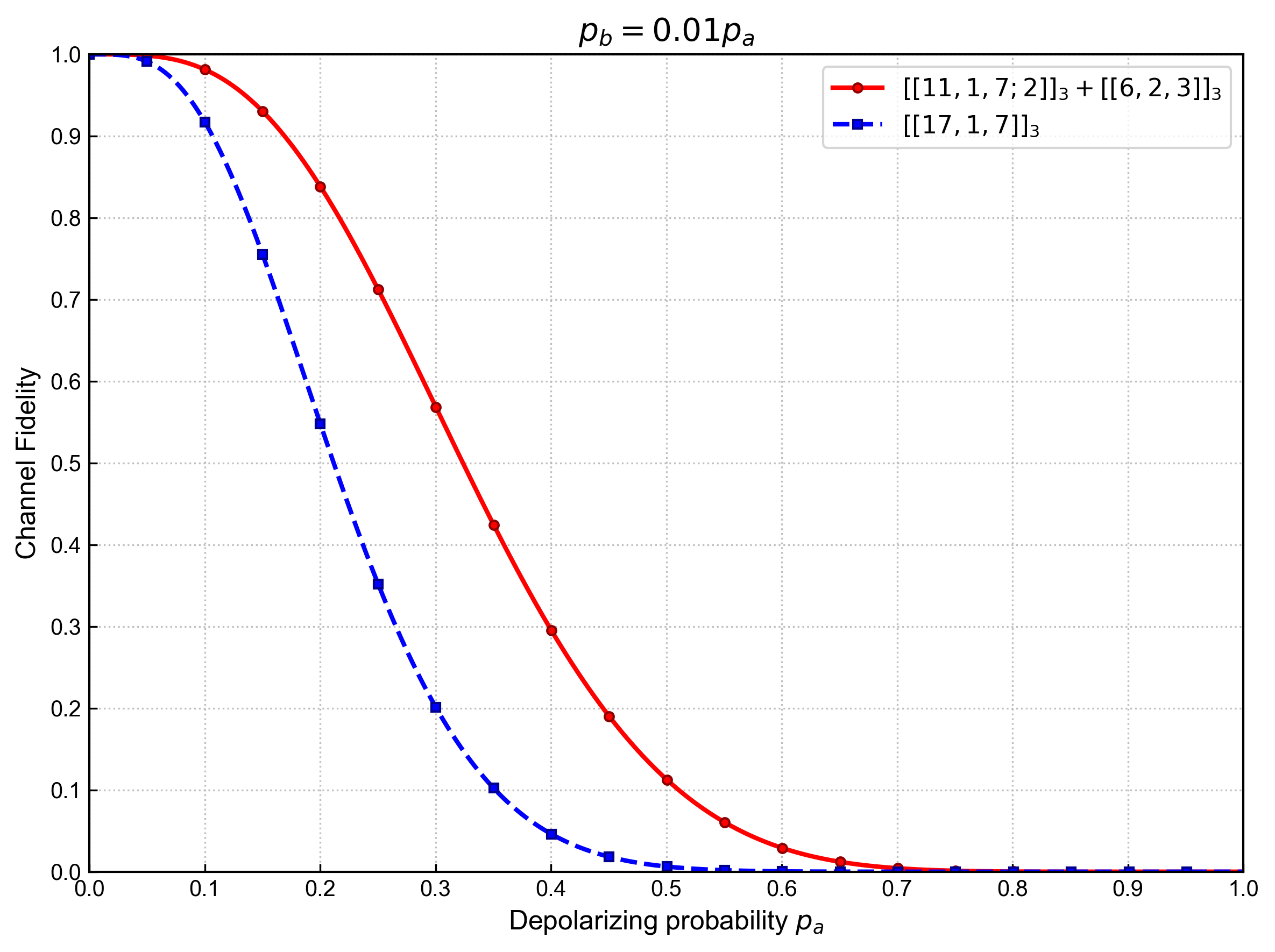} &
        \includegraphics[width=0.4\textwidth]{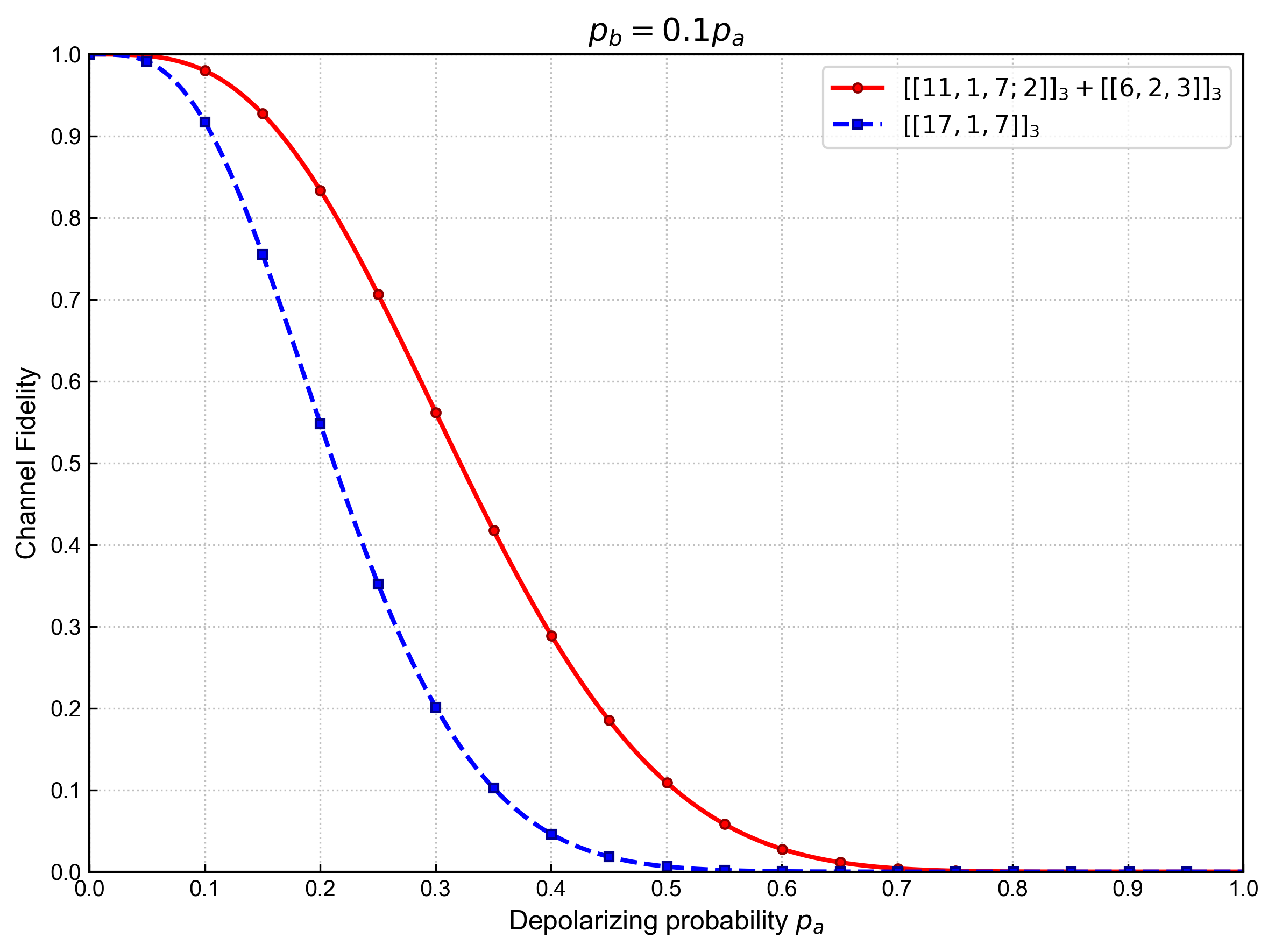} \\
        (a) $p_{b}=0.01 p_{a}$ & (b)$ p_{b}=0.1 p_{a}$ \\[1ex]
        \includegraphics[width=0.4\textwidth]{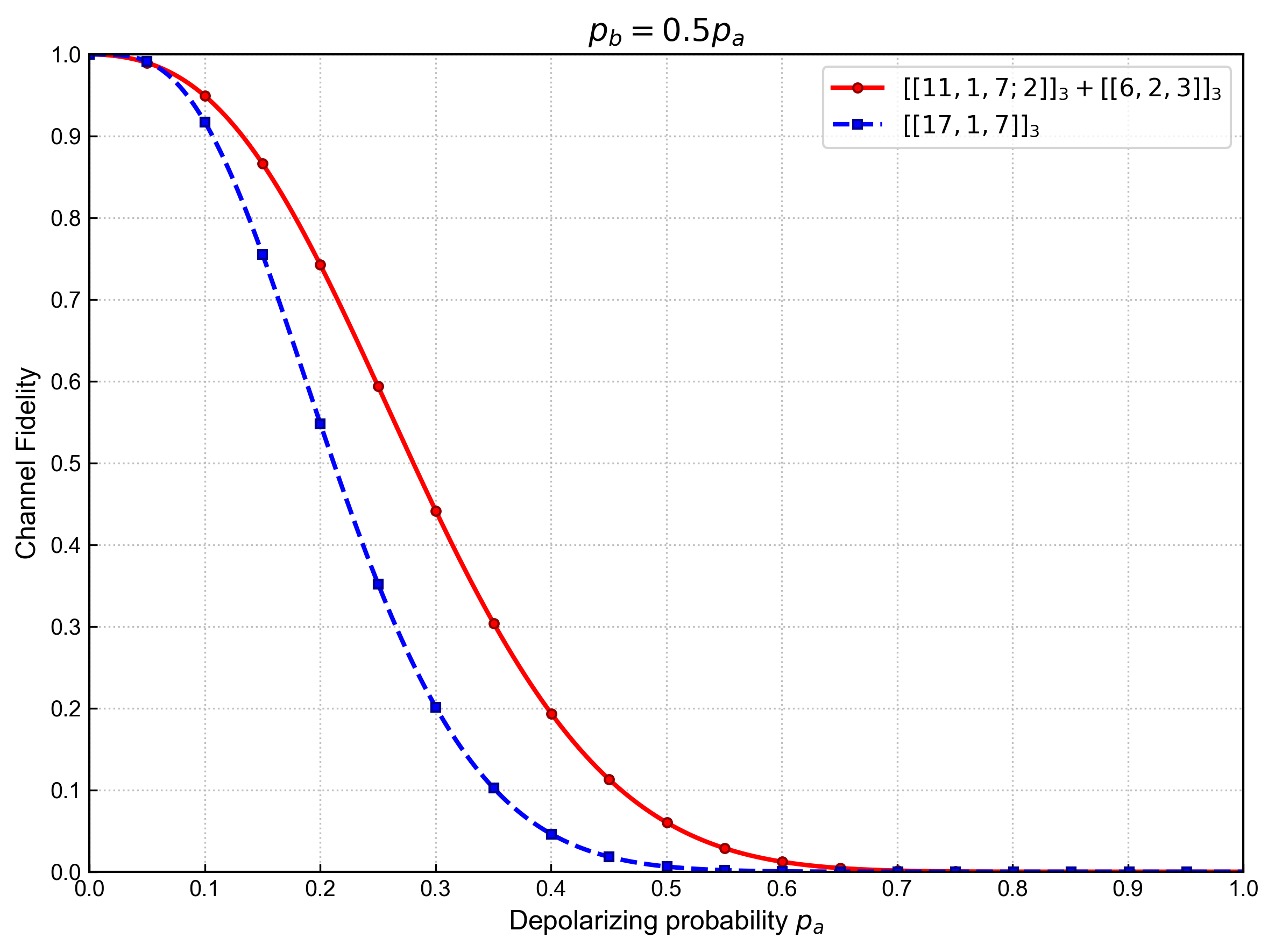} &
        \includegraphics[width=0.4\textwidth]{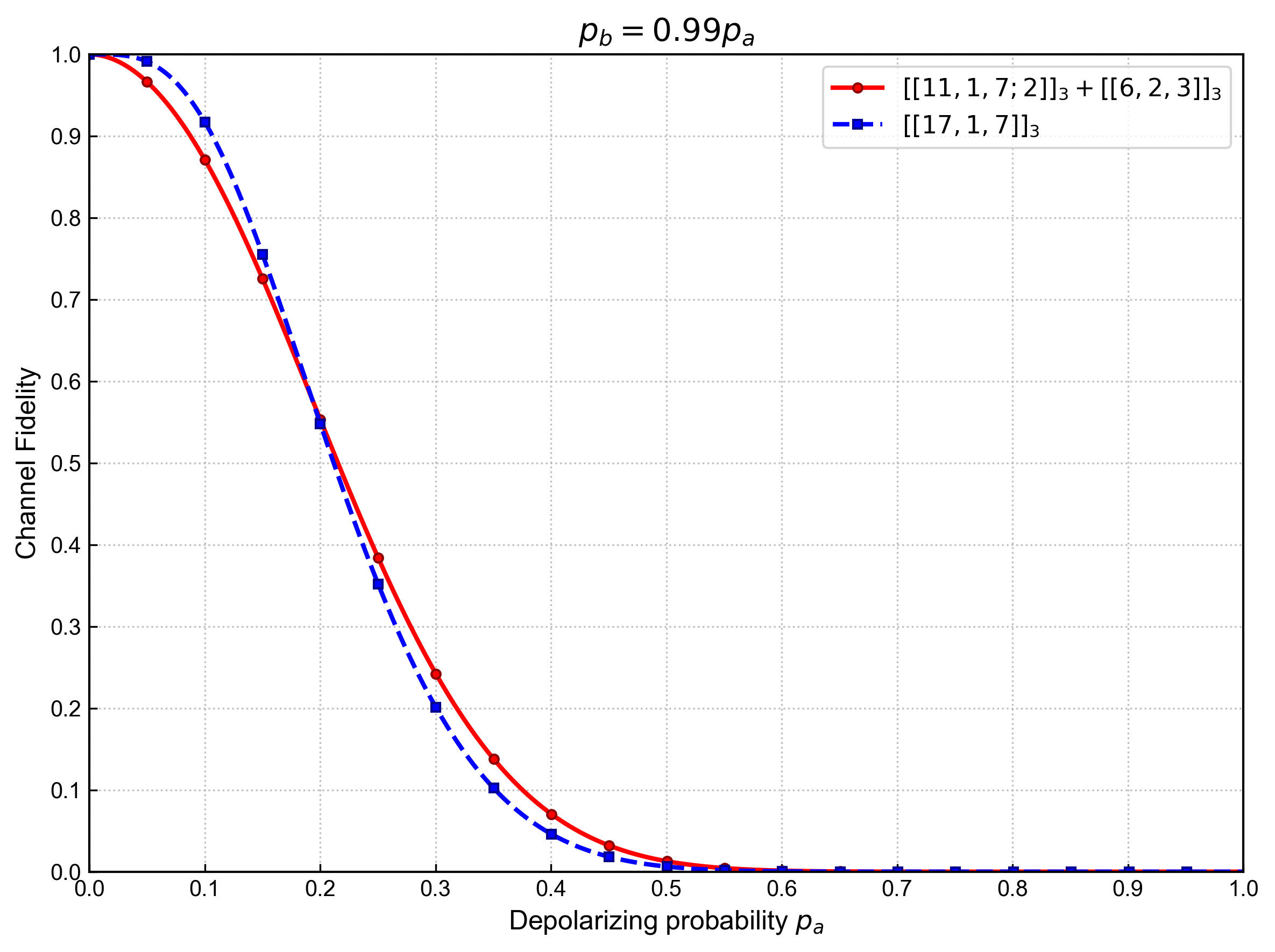} \\
        (c) $p_{b}=0.5 p_{a}$ & (d) $p_{b}=0.99 p_{a}$ \\
    \end{tabular}
    \caption{{Comparative analysis of the Channel fidelity under four scenarios.
    We use the depolarizing rate $p_a$ and the Channel fidelity as the $X$ and $Y$ axes, respectively.}  }
\end{figure}

It can be seen that as long as the selected entanglement degradation coefficient $\lambda= p_{b}/p_{a}$ is sufficiently small, that is, when the error probability of ebits is lower than that of qubits, the performance of
Ie-EAQECCs can be guaranteed to be superior to that of quantum stabilizer codes with corresponding parameters.

\section{\bf Conclusion}

We have developed a comprehensive stabilizer formalism for EAQECCs with noisy ebits over general finite fields $\mathbb{F}_q$. Our framework generalizes the binary constructions and provides a unified approach to
constructing EAQECCs-Ne using symplectic geometry and additive codes over $\mathbb{F}_{q^2}$.

We have presented explicit constructions of EAQECCs-Ne with good parameters and demonstrated that some of these codes outperform optimal standard stabilizer codes when the noise on the ebits is sufficiently small. This makes
EAQECCs-Ne a promising approach for practical quantum communication systems where both the transmitted qudits and the shared ebits are subject to noise.

\end{document}